\documentclass[DIV12,11pt,a4paper]{scrartcl}

\usepackage{tikz}
\usepackage{enumerate,amsmath,amsfonts,cite}
\usepackage{amsthm}
\usepackage{authblk}

\setkomafont{title}{\normalfont \LARGE \boldmath}
\setkomafont{section}{\normalfont \Large \bfseries \boldmath \boldmath}
\setkomafont{subsection}{\normalfont \large \bfseries \boldmath}
\setkomafont{subsubsection}{\normalfont \bfseries \boldmath}
\setkomafont{paragraph}{\normalfont \bfseries \boldmath}

\usetikzlibrary{calc}
\usetikzlibrary{intersections}

\newtheorem{claim}{Claim}[section]
\newtheorem{theorem}[claim]{Theorem}
\newtheorem{lemma}[claim]{Lemma}
\newtheorem{remark}[claim]{Remark}
\newtheorem{definition}[claim]{Definition}
\newtheorem{corollary}[claim]{Corollary}

\DeclareMathOperator{\probab}{\mathbb{P}}
\DeclareMathOperator{\expected}{\mathbb{E}}
\DeclareMathOperator{\pa}{\mathsf{PA}}
\DeclareMathOperator{\penergy}{\mathsf{pa}}
\DeclareMathOperator{\pt}{\mathsf{PT}}
\DeclareMathOperator{\heavy}{\mathsf{H}}
\DeclareMathOperator{\light}{\mathsf{L}}
\DeclareMathOperator{\reallight}{\mathsf{A}}
\DeclareMathOperator{\mst}{\mathsf{MST}}

\DeclareMathOperator{\functional}{\mathsf{F}}
\DeclareMathOperator{\diam}{diam}

\newcommand{\anglealpha}{30}

\newcommand{\meven}{M_{\operatorname{even}}}
\newcommand{\modd}{M_{\operatorname{odd}}}
\newcommand{\peven}{P_{\operatorname{even}}}
\newcommand{\podd}{P_{\operatorname{odd}}}
\newcommand{\eps}{\varepsilon}

\newcommand{\cgood}{c_{\textup{good}}}
\newcommand{\cbad}{c_{\textup{bad}}}

\newcommand{\stdexp}{\frac{d-p}d}

\newcommand{\nat}{\mathbb{N}}
\newcommand{\real}{\mathbb{R}}

\newcommand{\cball}{\ensuremath{c_{\textup{ball}}}}
\newcommand{\rball}{\ensuremath{r_{\textup{ball}}}}
\newcommand{\cedge}{\ensuremath{c_{\textup{edge}}}}
\newcommand{\redge}{\ensuremath{r_{\textup{edge}}}}
\newcommand{\rbig}{\ensuremath{r_{\textup{big}}}}

\title{Probabilistic Analysis of Power Assignments}

\author[1,2]{Maurits de Graaf}
\author[1]{Bodo Manthey}

\affil[1]{University of Twente, Department of Applied Mathematics, Enschede, Netherlands \authorcr \texttt{m.degraaf/b.manthey@utwente.nl}}
\affil[2]{Thales Nederland B.\ V., Huizen, Netherlands}

\begin{document}
\maketitle

\begin{abstract}
A fundamental problem for wireless ad hoc networks is the assignment of suitable transmission powers to
the wireless devices such that the resulting communication graph is connected. The goal is to minimize the total transmit power in order
to maximize the life-time of the network.
Our aim is a probabilistic analysis of this power assignment problem. We prove
complete convergence for arbitrary combinations of the dimension $d$ and the
distance-power gradient $p$. Furthermore, we prove that the expected approximation ratio of the
simple spanning tree heuristic is strictly less than its worst-case ratio of $2$.

Our main technical novelties are two-fold: First, we find a way to deal with the unbounded degree that
the communication network induced by the optimal power assignment can have.
Minimum spanning trees and traveling salesman tours, for which strong
concentration results are known in Euclidean space, have bounded degree, which is heavily exploited in their analysis.
Second, we apply a recent generalization of Azuma-Hoeff\-ding's inequality to prove complete convergence for the case $p \geq d$
for both power assignments and minimum spanning trees (MSTs).
As far as we are aware, complete convergence for $p > d$ has not been proved yet for any Euclidean functional.
\end{abstract}

\section{Introduction}

Wireless ad hoc networks have received significant attention
due to their many applications in, for instance, environmental monitoring or emergency disaster relief, where wiring
is difficult. Unlike wired networks, wireless ad hoc networks lack a backbone infrastructure. Communication takes place either through
single-hop transmission or by relaying through intermediate
nodes. We consider the
case that each node can
adjust its transmit power for the purpose of power conservation. In the assignment of transmit powers,
two conflicting effects have to be taken into account: if the
transmit powers are too low, the resulting network may be disconnected.
If the transmit powers are too high, the nodes run out of energy
quickly. The goal of the power assignment problem is to assign
transmit powers to the transceivers such that the resulting
network is connected and the sum of transmit powers is minimized~\cite{LloydEA}.

\subsection{Problem Statement and Previous Results}

We consider a set of vertices $X \subseteq [0,1]^d$, which represent the sensors, $|X|=n$, and assume
that $\|u-v\|^p$, for some $p \in \real$ (called the \emph{distance-power gradient} or \emph{path loss exponent}),
is the power required to successfully transmit a
signal from $u$ to $v$.
This is called the power-attenuation model, where the strength of the signal decreases with $1/r^p$ for distance $r$, and is a simple
yet very common model for power assignments in wireless networks~\cite{Rappaport:Wireless:2002}.
In practice, we typically have $1 \leq p \leq
6$~\cite{Pahlavan:OneSix:1995}.

A power assignment $\penergy: X \to [0, \infty)$ is an assignment of transmit powers
to the nodes in $X$.
Given $\penergy$, we have an edge between two nodes $u$ and $v$ if both $\penergy(x), \penergy(y)
\geq \|x-y\|^p$. If the resulting graph is connected, we call it a \emph{PA graph}.
Our goal is to find a PA graph and a corresponding power assignment $\penergy$ that
minimizes $\sum_{v \in X} \penergy(v)$. Note that any PA graph $G = (X,E)$ induces a power assignment
by $\penergy(v) = \max_{u \in X: \{u,v\} \in E} \|u-v\|^p$.

PA graphs can in many aspects be
regarded as a tree as we are only interested in connectedness, but it can contain more edges in general.
However, we can simply ignore edges and restrict ourselves to a spanning tree of the PA graph.

The minimal connected power assignment problem is NP-hard for $d
\geq 2$ and APX-hard for $d \geq
3$~\cite{ClementiEA:PowerRadio:2004}. For $d=1$, i.e., when the
sensors are located on a line, the problem can be solved by
dynamic programming~\cite{KirousisEA}. A simple approximation
algorithm for minimum power assignments is the minimum spanning
tree heuristic (MST heuristic), which achieves a tight worst-case
approximation ratio of $2$~\cite{KirousisEA}. This has been
improved by Althaus et al.~\cite{AlthausEA:RangeAssignment:2006},
who devised an approximation algorithm that achieves an
approximation ratio of $5/3$. A first average-case analysis of
the MST heuristic was presented
by de Graaf et al.~\cite{AveragePA}: First, they analyzed the expected approximation ratio
of the MST heuristic for the (non-geometric, non-metric) case of independent edge lengths.
Second, they proved convergence of the total power consumption of the assignment
computed by the MST heuristic for the special case of $p = d$, but not of the optimal power assignment.
They left as open problems, first, an average-case analysis of the MST heuristic
for random geometric instances and, second, the convergence of the value of the optimal power assignment.

Other power assignment problems studied include the $k$-station network coverage
problem of Funke et al.~\cite{FunkeEA:PowerTSP:2011}, where transmit powers are
assigned to at most $k$ stations such that $X$ can be reached from
at least one sender, or power assignments in the SINR model~\cite{HalldorssonEA,Kesselheim:SINR:2011}.

\subsection{Our Contribution}

In this paper, we conduct an average-case analysis of the optimal power assignment problem
for Euclidean instances. The points are drawn independently and uniformly
from the $d$-dimensional unit cube $[0,1]^d$. We believe that probabilistic analysis is a better-suited
measure for performance evaluation in wireless ad hoc networks, as the positions of the sensors -- in particular if deployed
in areas that are difficult to access -- are naturally random.

Roughly speaking, our contributions are as follows:
\begin{enumerate}
\setlength{\itemsep}{0mm}
\item We show that the power assignment functional has
sufficiently nice properties in order to apply Yukich's general framework
for Euclidean functionals~\cite{Yukich:ProbEuclidean:1998} to
obtain concentration results (Section~\ref{sec:properties}).

\item Combining these insights with
a recent generalization of the Azuma-Hoeff\-ding
bound~\cite{Warnke}, we obtain concentration of measure and complete convergence for all combinations of $d$ and $p \geq 1$,
 even for the case $p \geq d$ (Section~\ref{sec:convergence}).
 In addition, we obtain complete convergence for $p \geq d$ for minimum-weight spanning trees.
 As far as we are aware,
complete convergence for $p \geq d$ has not been proved yet for
such functionals. The only exception we are aware of are minimum spanning trees for the case $p=d$~\cite[Sect.~6.4]{Yukich:ProbEuclidean:1998}.

\item We provide a probabilistic analysis of the MST heuristic for the geometric case. We show that
  its expected approximation ratio is strictly smaller than its worst-case approximation ratio of $2$~\cite{KirousisEA} for any $d$ and $p$ (Section~\ref{sec:MST}).
\end{enumerate}
Our main technical contributions are two-fold: First, we introduce
a transmit power redistribution argument to deal with the
unbounded degree that graphs induced by the optimal transmit power
assignment can have. The unboundedness of the degree makes the
analysis of the power assignment functional $\pa$
challenging.
The reason is that removing a vertex can
cause the graph to fall into a large number of components and it might be costly to connect these
components without the removed vertex. In contrast, the degree of any minimum spanning tree,
for which strong concentration results are known in
Euclidean space for $p \leq d$, is bounded for every fixed $d$, and this is heavily
exploited in the analysis.
(The concentration result by de Graaf et al.~\cite{AveragePA}
for the power assignment obtained from the MST heuristic also exploits 
that MSTs have bounded degree.)

Second, we apply a recent
generalization of Azuma-Hoeffding's inequality by Warnke~\cite{Warnke} to
prove complete convergence for the case $p \geq d$ for both power
assignments and minimum spanning trees. We introduce the notion
of \emph{typically smooth} Euclidean functionals,
prove convergence of such functionals, and show that minimum spanning trees
and power assignments are typically smooth.
In this sense, our proof of complete convergence provides an alternative and generic way to prove complete convergence, whereas Yukich's proof
for minimum spanning trees is tailored to the case $p=d$. In order to prove complete convergence with our approach,
one only needs to prove convergence in mean, which is often much simpler than complete convergence, and typically smoothness.
Thus, we provide a simple method to prove complete convergence of Euclidean functionals along the lines of Yukich's result
that, in the presence of concentration of measure,
convergence in mean implies complete convergence~\cite[Cor.~6.4]{Yukich:ProbEuclidean:1998}.

\section{Definitions and  Notation}
\label{sec:def}

Throughout the paper, $d$ (the dimension) and $p$ (the distance-power gradient) are fixed constants.
For three points $x, y, v$, we by $\overline{xv}$ the line through
$x$ and $v$, and we denote by $\angle(x,v,y)$ the angle between $\overline{xv}$ and $\overline{yv}$.

A \emph{Euclidean functional} is a function $\functional^p$ for $p > 0$ that maps finite sets of points
in $[0,1]^d$ to some non-negative real number and is translation invariant and homogeneous of order $p$~\cite[page 18]{Yukich:ProbEuclidean:1998}.
From now on, we omit the superscript $p$ of Euclidean functionals, as $p$ is always fixed and clear from the context.

$\pa_B$ is the canonical boundary functional of $\pa$ (we refer to Yukich~\cite{Yukich:ProbEuclidean:1998} for
boundary functionals of other optimization problems):
given a
hyperrectangle $R \subseteq \real^d$ with $X \subseteq R$, this means that a solution is
an assignment $\penergy(x)$ of power to the nodes $x \in X$ such
that
\begin{itemize}
\setlength{\itemsep}{0mm}
\item $x$ and $y$ are connected if $\penergy(x), \penergy(y) \geq \|x-y\|^p$,
\item $x$ is connected to the boundary of $R$ if the distance of $x$ to the boundary of $R$ is at
most $\penergy(x)^{1/p}$, and
\item the resulting graph, called a \emph{boundary PA graph}, is either connected or consists of connected components
that are all connected to the boundary.
\end{itemize}
Then $\pa_B(X, R)$ is the minimum value for $\sum_{x \in X} \penergy(x)$ that can be achieved by a boundary PA graph.
Note that in the boundary
functional, no power is assigned to the boundary. 
It is straight-forward to see that $\pa$ and $\pa_B$ are Euclidean functionals for all $p> 0$ according to Yukich~\cite[page 18]{Yukich:ProbEuclidean:1998}.

For a hyperrectangle $R \subseteq \real^d$, let $\diam R = \max_{x,y \in R} \|x-y\|$ denote the diameter of $R$.
For a Euclidean functional $\functional$, let $\functional(n) =
\functional(\{U_1,\ldots, U_n\})$, where $U_1, \ldots, U_n$ are drawn uniformly and independently from $[0,1]^d$.
Let
\[
\gamma_{\functional}^{d,p} = \lim_{n \to \infty} \frac{\expected\bigl(\functional(n)\bigr)}{n^\stdexp}.
\]
(In principle, $\gamma_{\functional}^{d,p}$ need not exist, but it does exist for all functionals
considered in this paper.)

A sequence $(R_n)_{n \in \nat}$
of random variables \emph{converges in mean} to a constant $\gamma$
if $\lim_{n \to \infty} \expected(|R_n - \gamma|)= 0$.
The sequence $(R_n)_{n \in \nat}$ \emph{converges completely to a constant $\gamma$} if
we have
\[
\sum_{n=1}^\infty \probab\bigl(|R_n-\gamma| > \eps\bigr) < \infty
\]
for all $\eps > 0$.

Besides $\pa$, we consider two other Euclidean functions:
$\mst(X)$ denotes the length of the minimum spanning tree with lengths raised to the power $p$.
$\pt(X)$ denotes the total power consumption of the assignment obtained from the MST heuristic, again with lengths raised
to the power $p$.
The MST heuristic proceeds as follows: First, we compute a minimum spanning tree of $X$. 
The let $\penergy(x) = \max\{\| x-y\|^p \mid \text{$\{x,y\}$ is an edge of the MST}\}$.
By construction and a simple analysis, we have $\mst(X) \leq \pa(X) \leq \pt(X) \leq 2 \cdot \mst(X)$~\cite{KirousisEA}.

For $n \in \nat$, let $[n] = \{1, \ldots, n\}$.

\section{Properties of the Power Assignment Functional}
\label{sec:properties}

After showing that optimal PA graphs can have unbounded degree and providing a lemma that helps solving this problem,
we show that the power assignment functional fits into Yukich's framework for Euclidean functionals~\cite{Yukich:ProbEuclidean:1998}.

\subsection{Degrees and Cones}
\label{ssec:cones}
As opposed to minimum spanning trees, whose maximum degree is bounded from above by a constant that depends only on the dimension $d$,
a technical challenge is that the maximum degree in an optimal PA graphs cannot be bounded by a constant in the dimension. This holds even for the simplest case of $d=1$ and $p > 1$. We conjecture that the same holds also for $p=1$, but proving
this seems to be more difficult and not to add much.

\begin{lemma}
\label{lem:unbounded} For all $p > 1$, all integers $d \geq 1$,
and for infinitely many $n$, there exists instances of $n$ points
in $[0,1]^d$ such that the unique optimal PA graph is a tree with
a maximum degree of $n-1$.
\end{lemma}

\begin{proof}
Let $n$ be odd, and let $2m+1=n$. Consider the instance
\[
  X_m = \{a_{-m}, a_{-m+1}, \ldots, a_0, \ldots, a_{m-1}, a_m\}
\]
that consists of $m$ positive integers $a_1, \ldots, a_m$, $m$ negative integers
$a_{-i} = -a_{i}$ for $1 \leq i \leq m$, and $a_0 = 0$.
We assume that $a_{i+1} \gg a_{i}$ for all $i$.
By scaling and shifting, we can achieve that $X$ fits into the unit interval.

A possible solution $\penergy: X_m \rightarrow \real^{+}$ is
assigning power $a_i^p$ to $a_i$ and $a_{-i}$ for $1 \leq i \leq
m$ and power $a_m^p$ to $0$. In this way, all points are connected
to $0$. We claim that this power assignment is the unique optimum.
As $a_m = -a_{-m} \gg |a_{i}|$ for $|i| < m$,
the dominant term in the power consumption $\Psi_m$ is $3 a_m^p$
(the power of $a_m$, $a_{-m}$, and $a_0 = 0$). Note that no other
term in the total power consumption involves $a_m$.

We show that $a_m$ and $a_{-m}$ must be connected to $0$ in an optimal PA graph.
First, assume that $a_m$ and $a_{-m}$ are connected to different vertices. Then the total power consumption increases
to about $4 a_m^p$ because $a_{\pm m}$ is very large compared to $a_i$ for all $|i| < m$
(we say that $a_m$ is dominant).
Second, assume that $a_m$ and $a_{-m}$ are connected to $a_i$ with $i \neq 0$. Without loss of generality, we assume that $i>0$ and, thus,
$a_i > 0$.
Then the total power consumption is at least
$2 \cdot (a_m + a_i)^p + (a_m - a_i)^p \geq 3a_m^p + 2a_m^{p-1}a_i$. Because $a_m$ is dominant, this is strictly
more than $\Psi_m$ because it contains the term $2a_m^{p-1}a_i$, which contains the very large $a_m$ because $p>1$.

From now on, we can assume that $0 = a_0$ is connected to $a_{\pm
m}$. Assume that there is some point $a_i$ that is connected to
some $a_j$ with $i, j \neq 0$. Assume without loss of generality
that $i > 0$ and $|i| \geq |j|$. Assume further that $i$ is
maximal in the sense that there is no $|k| > i$ such that $a_k$ is
connected to some vertex other than $0$. We set $a_i$'s power to
$a_i^p$ and $a_j$'s power to $|a_j|^p$. Then both are connected to
$0$ as $0$ has already sufficient power to send to both.
Furthermore, the PA graph is still connected: All vertices $a_k$
with $|k| > i$ are connected to $0$ by the choice of $i$. If some
$a_k$ with $|k| \leq i$ and $k \neq i, j$ was connected to $a_i$
before, then it has also sufficient power to send to $0$.

The power balance remains to be considered: If $j = -i$, then the
energy of both $a_i$ and $a_j$ has been strictly decreased.
Otherwise, $|j| < i$. The power of $a_i$ was at least
$(a_i-a_j)^p$ before and is now $a_i^p$. The power of $a_j$ was at
least $(a_i-a_j)^p$ before and is now $a_j^p$. Since $a_i$
dominates all $a_j$ for $|j| < i$, this decreases the power.
\end{proof}

The unboundedness of the degree of PA graphs make the analysis of the functional $\pa$ challenging. The technical reason is
that removing a vertex can cause the PA graph to fall into a non-constant number of components.
The following lemma is the crucial ingredient to get over this ``degree hurdle''.

\begin{lemma}
\label{lem:conefactor}
Let $x, y \in X$, let $v \in [0,1]^d$, and assume that $x$ and $y$ have power $\penergy(x) \geq \|x-v\|^p$ and $\penergy(y) \geq \|y-v\|^p$,
respectively. Assume further that $\|x - v\| \leq \|y-v\|$
and that $\angle(x,v,y) \leq \alpha$ with $\alpha \leq \pi/3$.
Then the following holds:
\begin{enumerate}[(a)]
\setlength{\itemsep}{0mm}
\item $\penergy(y) \geq \|x-y\|^p$, i.e., $y$ has sufficient power to reach $x$.
\label{reachcloser}
\item If $x$ and $y$ are not connected (i.e., $\penergy(x)< \|x - y \|^p$), then $\|y - v\| > \frac{\sin(2\alpha)}{\sin(\alpha)} \cdot \|x-v\|$.
\label{sqrtthree}
\end{enumerate}
\end{lemma}

\begin{proof}
Because $\alpha \leq \pi/3$, we have $\|y-v\| \geq \|y-x\|$.
This implies~\eqref{reachcloser}.

The point $x$ has sufficient power to reach any point within a
radius of $\|x-v\|$ of itself. By~\eqref{reachcloser}, point $y$
has sufficient power to send to $x$. Thus, if $y$ is within a
distance of $\|x-v\|$ of $x$, then also $x$ can send to $y$ and,
thus, $x$ and $y$ are connected. We project $x$, $y$, and $v$ into
the two-dimensional subspace spanned by the vectors $x-v$ and
$y-v$. This yields a situation as depicted in
Figure~\ref{fig:conefactor}. Since $\penergy(x) \geq \|x-v\|^p$, point
$x$ can send to all points in the light-gray region, thus in
particular to all dark-gray points in the cone rooted at $v$. In
particular, $x$ can send to all points that are no further away
from $v$ than the point $z$. The triangle $vxz$ is isosceles.
Thus, also the angle at $z$ is $\alpha$ and the angle at $x$ is
$\beta = \pi - 2\alpha$. Using the law of sines together
with $\sin(\beta) = \sin(2\alpha)$ yields that $\|z-v\| =
\frac{\sin(2\alpha)}{\sin(\alpha)} \cdot \|x-v\|$, which completes
the proof of \eqref{sqrtthree}.
\end{proof}

\begin{figure}[t]
\centering
\begin{tikzpicture}[scale=2]
\tikzstyle{Thickness}=[line width=0.8pt]
\coordinate (v) at (0,0);
\coordinate (x) at (0:1);
\path[name path=xcircle] (x) circle (1cm);
\fill[white!85!black] (x) circle (1cm);

\path[name path=rightborder] (-\anglealpha:0.1) -- (-\anglealpha:2);
\path[name path=leftborder] (\anglealpha:0.1) -- (\anglealpha:2);
\path[-, name intersections={of=rightborder and xcircle, by=rightfarthest}];
\path[-, name intersections={of=leftborder and xcircle, by=leftfarthest}];

\fill[white!70!black] (v) -- (rightfarthest) arc (-2*\anglealpha:2*\anglealpha:1) -- cycle;

\draw[dotted, Thickness] (rightfarthest) arc (-\anglealpha:\anglealpha:{sin(2*\anglealpha)/sin(\anglealpha)});

\draw[Thickness] (-\anglealpha:2) -- (v) -- (\anglealpha:2);
\draw[Thickness] (v) -- (0:2.1);
\draw[Thickness] (-\anglealpha:0.4) arc (-\anglealpha:\anglealpha:0.4);
\draw[Thickness] ($(rightfarthest)+(180-2*\anglealpha:0.4)$) arc (180-2*\anglealpha:180-\anglealpha:0.4);
\draw[Thickness] ($(x)+(180:0.28)$) arc (180:360-2*\anglealpha:0.28);

\node at (-0.5*\anglealpha:0.3) {$\alpha$};
\node at (0.5*\anglealpha:0.3) {$\alpha$};
\node at ($(rightfarthest)+(180-1.5*\anglealpha:0.3)$) {$\alpha$};
\node at ($(x)+(270-\anglealpha:0.16)$) {$\beta$};

\draw[Thickness, dashed] (leftfarthest) -- (x) -- (rightfarthest);

\tikzstyle{Node}=[circle, fill=black, inner sep=0pt, minimum size=2mm]
\node[Node] (nv) at (v) {};
\node[left] at (nv) {$v$};
\node[Node] (nx) at (x) {};
\node[above] at (nx) {$x$};
\node[Node] (nz) at (rightfarthest) {};
\node[below] at (rightfarthest) {$z$};
\end{tikzpicture}
\caption{Point $x$ can send to all points in the gray area as it can send to $v$. In particular, $x$ can send to all points that are no further away from $v$ than $z$.
This includes all points to the left of the dotted line. The dotted line consists of points at a distance
of $\frac{\sin(2\alpha)}{\sin(\alpha)} \cdot \|x-v\|$ of $v$.}
\label{fig:conefactor}
\end{figure}
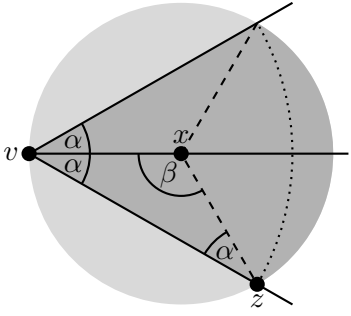

For instance, $\alpha = \pi/6$ results in a factor of
$\sqrt{3} = \sin(\pi/3)/\sin(\pi/6)$. In the following, we
invoke this lemma always with $\alpha = \pi/6$, but this
choice is arbitrary as long as $\alpha < \pi/3$, which causes
$\sin(2\alpha)/\sin(\alpha)$ to be strictly larger than $1$.

\subsection{Deterministic Properties}
\label{ssec:deterministic}

In this section, we state properties of the power assignment functional.
Subadditivity (Lem\-ma~\ref{lem:subadd}), superadditivity (Lemma~\ref{lem:supadd}),
and growth bound (Lem\-ma~\ref{lem:growth}) are straightforward.

\begin{lemma}[subadditivity]
\label{lem:subadd}
$\pa$ is subadditive~\cite[(2.2)]{Yukich:ProbEuclidean:1998} for all $p>0$ and all $d \geq 1$,
i.e., for any point sets $X$ and $Y$ and any hyperrectangle $R\subseteq \real^d$ with $X, Y \subseteq R$, we have
\[
\pa(X \cup Y) \leq \pa(X) + \pa(Y) + O\bigl((\diam R)^p\bigr).
\]
\end{lemma}

\begin{proof}
Let $T_X$ and $T_Y$ be optimal PA graphs for $X$ and $Y$, respectively.
We connect these graphs by an edge of length at most $\diam R$.
This yields a solution for $X \cup Y$, i.e., a PA graph, and the additional costs are
bounded from above by the length of this edge to the power $p$, which is
bounded by $(\diam R)^p$.
\end{proof}

\begin{lemma}[superadditivity]
\label{lem:supadd}
$\pa_B$ is superadditive for all $p \geq 1$ and $d \geq 1$~\cite[(3.3)]{Yukich:ProbEuclidean:1998}, i.e.,
for any $X$, hyperrectangle $R\subseteq \real^d$ with $X \subseteq R$ and partition of
$R$ into hyperrectangles $R_1$ and $R_2$, we have
\[
   \pa_B^p(X, R) \geq \pa^p_B(X \cap R_1, R_1) + \pa^p_B(X \cap R_2, R_2) .
\]
\end{lemma}

\begin{proof}
Let $T$ be an optimal boundary PA graph for $(X, R)$.
This graph restricted to $R_1$ and $R_2$ yields
boundary graphs $T_1$ and $T_2$ for $(X \cap R_1, R_1)$ and $(X \cap R_2, R_2)$, respectively.
The sum of the costs of $T_1$ and $T_2$ is upper bounded by
the costs of $T$ because $p \geq 1$ (splitting an edge at the border between
$R_1$ and $R_2$ results in two edges whose sum of lengths to the power $p$
is at most the length of the original edge to the power $p$).
\end{proof}

\begin{lemma}[growth bound]
\label{lem:growth}
For any $X \subseteq [0,1]^d$ and $0 < p$ and $d \geq 1$, we have
\[
\pa_B(X) \leq \pa(X) \leq O\left(\max\left\{n^{\stdexp}, 1\right\}\right).
\]
\end{lemma}

\begin{proof}
This follows from the growth bound for the MST~\cite[(3.7)]{Yukich:ProbEuclidean:1998}, because
$\mst(X) \leq \pa(X) \leq 2\mst(X)$ for all $X$~\cite{KirousisEA}.
The inequality $\pa_B(X) \leq \pa(X)$ holds obviously.
\end{proof}

The following lemma shows that $\pa$ is smooth, which roughly means that adding or removing a few points
does not have a huge impact on the function value. Its proof requires
Lemma~\ref{lem:conefactor} to deal with the fact that optimal PA graphs can have unbounded degree.

\begin{lemma}
\label{lem:smooth}
The power assignment functional
$\pa$ is smooth for all $0 < p \leq d$~\cite[(3.8)]{Yukich:ProbEuclidean:1998},
i.e.,
\[
\bigl| \pa^p(X \cup Y) - \pa^p(X)\bigr|  = O\left(|Y|^\stdexp\right)
\]
for all point sets $X, Y \subseteq [0,1]^d$.
\end{lemma}

\begin{proof}
One direction is straightforward:
$\pa(X \cup Y) - \pa(X)$ is bounded by $\Psi = O\bigl(|Y|^\stdexp\bigr)$, because
the optimal PA graph for $Y$ has a value of at most $\Psi$ by
Lemma~\ref{lem:growth}.
Then we can take the PA graph for $Y$ and connect it to
the tree for $X$ with a single edge, which costs at most $O(1) \leq \Psi$
because $p \leq d$.

For the other direction, consider the optimal PA graph $T$ for $X
\cup Y$. The problem is that the degrees $\deg_T(v)$ of vertices
$v \in Y$ can be unbounded (Lemma~\ref{lem:unbounded}). (If the
maximum degree were bounded, then we could argue in the same way
as for the MST functional.) The idea is to exploit the fact that
removing $v \in Y$ also frees some power. Roughly speaking,
we proceed as follows: Let $v \in Y$ be a vertex of possibly
large degree. We add the power of $v$ to some vertices close to
$v$. The graph obtained from removing $v$ and distributing its
energy has only a constant number of components.

To prove this, Lemma~\ref{lem:conefactor} is crucial.
We consider cones rooted at $v$ with the following properties:
\begin{itemize}
\setlength{\itemsep}{0mm}
\item The cones have a small angle $\alpha$, meaning that for every cone $C$ and every
   $x, y \in C$, we have $\angle(x,v,y) \leq \alpha$. We choose $\alpha = \pi/6$.
\item Every point in $[0,1]^d$ is covered by some cone.
\item There is a finite number of cones. (This can be achieved because $d$ is a constant.)
\end{itemize}

Let $C_1, \ldots, C_m$ be these cones. By abusing notation, let $C_i$ also denote all points $x \in C_i \cap (X \cup Y \setminus \{v\})$
that are adjacent to $v$ in $T$.
For $C_i$, let $x_i$ be the point in $C_i$ that is closest to $v$ and adjacent to $v$ (breaking ties arbitrarily), and let $y_i$ be the point
in $C_i$ that is farthest from $v$ and adjacent to $v$ (again breaking ties arbitrarily). (For completeness,
we remark that then $C_i$ can be ignored if $C_i \cap X = \emptyset$.) Let $\ell_i = \|y_i-v\|$ be the maximum distance of
any point in $C_i$ to $v$, and let $\ell = \max_i \ell_i$.

We increase the power of $x_i$ by $\ell^p/m$. Since the power of
$v$ is at least $\ell^p$ and we have $m$ cones, we can account for
this with $v$'s power because we remove $v$. Because $\alpha = \pi/6$ and $x_i$ is
closest to $v$, any point in $C_i$ is closer to $x_i$ than to $v$.
According to Lemma~\ref{lem:conefactor}\eqref{reachcloser}, every
point in $C_i$ has sufficient power to reach $x_i$. Thus, if $x_i$
can reach a point $z \in C_i$, then there is an established
connection between them.

From this and increasing $x_i$'s power to at least $\ell^p/m$,
there is an edge between $x_i$ and every point $z \in C_i$ that
has a distance of at most $\ell/\sqrt[p]m$ from $v$. We recall
that $m$ and $p$ are constants.

Now let $z_1, \ldots, z_k \in C_i$ be the vertices in $C_i$ that
are not connected to $x_i$ because $x_i$ has too little power. We
assume that they are sorted by increasing distance from $v$. Thus,
$z_k = y_i$.
We can assume that no two $z_j$ and $z_{j'}$ are in
the same component after removal of $v$. Otherwise, we can simply
ignore one of the edges $\{v, z_j\}$ and $\{v, z_{j'}\}$ without
changing the components.

Since $z_j$ and $z_{j+1}$ were connected to $v$ and they are not connected to each other,
we can apply Lemma~\ref{lem:conefactor}\eqref{sqrtthree}, which implies that  $\|z_{j+1} - v\| \geq \sqrt 3 \cdot \|z_j  - v\|$.
Furthermore, $\|z_1-v\| \geq \ell/\sqrt[p]m$ by assumption.
Iterating this argument yields $\ell = \|z_k - v\| \geq \sqrt{3}^{k-1} \|z_1 -v\| \geq \sqrt{3}^{k-1} \cdot \ell/\sqrt[p]m$.
This implies $k \leq \log_{\sqrt 3}(\sqrt[p]m) +1$. Thus, removing $v$ and redistributing its energy as described causes the PA graph to fall into at most
a constant number of components. Removing $|Y|$ points causes the PA graph to fall into at most $O(|Y|)$ components.
These components can be connected with costs $O(|Y|^{\stdexp})$ by choosing one point per component and applying Lemma~\ref{lem:growth}.
\end{proof}

\begin{lemma}
\label{lem:smoothboundary}
$\pa_B$ is smooth for all $1 \leq p \leq d$~\cite[(3.8)]{Yukich:ProbEuclidean:1998}.
\end{lemma}

\begin{proof}
The idea is essentially identical to the proof of Lemma~\ref{lem:smooth},
and we use the same notation.
Again, one direction is easy.
For the other direction, note that every vertex of $G=(X,E)$, with $E$ induced by $\penergy$ is connected to at most
one point at the boundary. We use the same kind of cones as for Lemma~\ref{lem:smooth}.
Let $v \in G$ be a vertex that we want to remove.
We ignore $v$'s possible connection to the boundary and proceed with the remaining
connections. In this way, we obtain a forest with $O(|G|)$ components.
We compute a boundary PA graph for one vertex of each component and are done because
of Lemma~\ref{lem:growth} and in the same way as in the proof of Lemma~\ref{lem:smooth}.
\end{proof}

Crucial for convergence of $\pa$ is that $\pa$, which is subadditive, and $\pa_B$, which is superadditive,
are close to each other. Then both are close to being both subadditive and superadditive.
The following lemma
states that indeed $\pa$ and $\pa_B$ do not differ too much for $1 \leq p < d$.

\begin{lemma}
\label{lem:pclose}
$\pa$  is point-wise close to $\pa_B$ for $1 \leq p < d$~\cite[(3.10)]{Yukich:ProbEuclidean:1998},
i.e.,
\[
 \bigl|\pa^p(X) - \pa^p_B(X, [0,1]^d)\bigr| = o\bigl(n^\stdexp\bigr)
\]
for every set $X \subseteq [0,1]^d$ of $n$ points.
\end{lemma}

\begin{proof}
Let $T$ be an optimal boundary PA graph for $X$.
Let $Q\subseteq X$ be the set of points that have a connection to
the boundary of $T$ and let $\partial Q$ be the corresponding points on the boundary. If we remove the connections to
the boundary,
we obtain a graph $T'$. We can assume
that $Q$ contains exactly one point per connected component of the graph $T'$.

We use the same dyadic decomposition as Yukich~\cite[proof of Lemma 3.8]{Yukich:ProbEuclidean:1998}.
This yields that the sum of transmit powers used to connect to the boundary is
bounded by the maximum of $O(n^{\frac{d-p-1}{d-1}})$ and $O(\log n)$
for $p \leq d-1$
and by a constant for $p \in (d-1, d)$.
We omit the proof as it is basically identical to Yukich's proof.

Let $Q \subseteq X$ be the points connected to the boundary, and let
$\partial Q$ be the points where $Q$ connects to the boundary.
We compute a minimum-weight spanning tree $Z$ of $\partial Q$.
(Note that we indeed compute an MST and not a PA. This is because the
MST has bounded degree and $\pa$ and $\mst$ differ by at most a factor of $2$.)
This MST $Z$ has a weight of
\[
O\left(\max\left\{n^{\frac{d-1-p}{d-1}}, 1\right\}\right) = o\left(n^{\stdexp}\right)
\]
according to the
growth bound for $\mst$~\cite[(3.7)]{Yukich:ProbEuclidean:1998}.
and because $d > p$.
If two points $\tilde q, \tilde q' \in \partial Q$ are connected
in this tree, then we connect the corresponding points $q, q' \in Q$.

The question that remains is by how much the power of the vertices
in $Q$ has to be increased in order to allow the connections as
described above. If $q, q' \in Q$ are connected, then an upper
bound for their power is given by the $p$-th power of their
distances to the boundary points $\tilde q$ and $\tilde q'$ plus
the length of the edge connecting $\tilde q$ and $\tilde q'$.
Applying the triangle inequality for powers of metrics twice, the
energy needed for connecting $q$ and $q'$ is at most $4^p=O(1)$
times the sum of these distances. Since the degree of $Z$ is
bounded, every vertex in $Q$ contributes to only a constant number
of edges and, thus, only to the power consumption of a constant
number of other vertices. Thus, the total additional power needed
is bounded by a constant times the power of connecting $Q$ to the
boundary plus the power to use $Z$ as a PA graph.
Because of the triangle inequality for powers of metrics, the bounded degree of every
vertex of $\partial Q$ in $Z$, and because of the dyadic decomposition mentioned
above, the increase of power is in compliance with the statement of the lemma.
\end{proof}

\begin{remark}
Lemma~\ref{lem:pclose} is an analogue of its counterpart for MST, TSP, and matching~\cite[Lemma 3.7]{Yukich:ProbEuclidean:1998} in terms
of the bounds. Namely, we obtain
\[
\bigl|\pa(X) - \pa_B(X)\bigr| \leq \begin{cases}
O(|X|^{\frac{d-p-1}{d-1}}) & \text{if $1 \leq p < d-1$},\\
O(\log|X|) & \text{if $p=d-1 \neq 1$},\\
O(1) & \text{if $d-1<p<d$ or $p=d-1=1$.}
\end{cases}
\]
\end{remark}

\subsection{Probabilistic Properties}
\label{ssec:probabilistic}

For $p > d$, smoothed is not guaranteed to hold, and for $p \geq d$, point-wise closeness is not guaranteed to hold.
But similar properties typically hold for random point sets,
namely smoothness in mean (Definition~\ref{def:smoothinmean}) and closeness in mean (Definition~\ref{def:closeinmean}).
In the following, let $X=\{U_1, \ldots, U_n\}$. Recall that $U_1, \ldots, U_n$ are drawn uniformly and independently from $[0,1]^d$.

Before proving smoothness in mean, we need a statement about the longest
edge in an optimal PA graph and boundary PA graph. The bound is asymptotically equal to the bound for the longest edge
in an MST~\cite{Penrose:LongestMST:1997,KozmaEA:ConnectivityThreshold:2010,GuptaKumar:CriticalPower:1999}.

To prove our bound for the longest edge in optimal PA graphs (Lemma~\ref{lem:longest}), we need the following
two lemmas. Lemma~\ref{lem:emptyball} is essentially equivalent to a result by Kozma et al.~\cite{KozmaEA:ConnectivityThreshold:2010},
but they do not state the probability explicitly.
Lemma~\ref{lem:iterclose} is a straight-forward consequence of Lemma~\ref{lem:emptyball}.
Variants of both lemmas are
known~\cite{Steele:ProbabilisticClassical:1990,Penrose:StrongMST:1999,Penrose:LongestMST:1997,
GuptaKumar:CriticalPower:1999}, but, for completeness, we state and prove both lemmas in the forms that we need.

\begin{lemma}
\label{lem:emptyball}
For every $\beta > 0$, there exists a
$\cball = \cball(\beta, d)$ such that, with a probability of at least
$1-n^{-\beta}$, every hyperball of radius $\rball = \cball \cdot
(\log n/n)^{1/d}$ and with center in $[0,1]^d$ contains at least
one point of $X$ in its interior.
\end{lemma}

\begin{proof}
We sketch the simple proof.
We cover $[0,1]^d$ with hypercubes of side length $\Omega(\rball)$
such that every ball of radius $\rball$ -- even if its center is
in a corner (for a point on the boundary, still at least a $2^{-d}
= \Theta(1)$ fraction is within $[0,1]^d$) -- contains at least
one box. The probability that such a box does not contain
a point, which is necessary for a ball to be empty, is at most
$\bigl(1-\Omega(\rball)^d\bigr)^n
\leq n^{-\Omega(1)}$ by independence of the points in $X$
and the definition of $\rball$.
The rest of the proof follows by a union bound over all
$O(n/\log n)$ boxes.
\end{proof}

We also need the following lemma, which essentially
states that if $z$ and $z'$ are sufficiently far away, then there
is -- with high probability -- always a point $y$ between $z$ and
$z'$ in the following sense: the distance of $y$ to $z$ is within
a predefined upper bound $2\rball$, and $y$ is closer to $z'$ than
$z$.

\begin{lemma}
\label{lem:iterclose}
For every $\beta > 0$, with a probability of at least $1- n^{-\beta}$, the following holds:
For every choice of $z, z' \in [0,1]^d$ with $\|z-z'\| \geq 2\rball$, there exists
a point $y \in X$ with the following properties:
\begin{itemize}
\item $\|z-y\| \leq 2\rball$.
\item $\|z'-y\| < \|z' - z\|$.
\end{itemize}
\end{lemma}

\begin{proof}
The set of candidates for $y$ contains a ball of radius $\rball$,
namely a ball of this radius whose center is at a distance
of $\rball$ from $z$ on the line between $z$ and $z'$.
This allows us to use Lemma~\ref{lem:emptyball}.
\end{proof}

\begin{lemma}[longest edge]
\label{lem:longest}
For every constant $\beta > 0$, there exists a
constant $\cedge=\cedge(\beta)$ such that, with a probability of at least
$1 - n^{-\beta}$, every edge of an optimal PA graph and an optimal
boundary PA graph $\pa_B$ is of length at most $\redge = \cedge
\cdot (\log n/n)^{1/d}$.
\end{lemma}

\begin{proof}
We restrict ourselves to considering PA graphs. The proof for boundary PA graphs
is almost identical.

Let $T$ be any PA graph. Let
$\cedge = 4 k^{1/p} \cball/(1-\sqrt{3}^{-p})^{1/p}$,
where $k$ is an upper bound for the number of vertices without a pairwise
connection at a distance between $r$ and $r/\sqrt{3}$ for
arbitrary $r$. It follows from Lemma~\ref{lem:conefactor} and its proof, that $k$ is a constant
that depends only on $p$ and $d$.

Note that $\cedge
> 2 \cball$. We are going to show that if $T$ contains an edge that is longer
than $\redge$, then we can find a better
PA graph with a probability of at least
$1 - n^{-\beta}$, which shows that $T$ is not optimal.

Let $v$ be a vertex incident to the
longest edge of $T$, and let $\rbig > \redge$ be the length of
this longest edge. (The longest edge is unique with a probability
of $1$. The node $v$ is not unique as the longest edge connects
two points.) We decrease the power of $v$ to $\rbig/\sqrt 3$. This
implies that $v$ loses contact to some points -- otherwise, the
power assignment was clearly not optimal.

The number $\cball$ depends on the exponent $\beta$ of the lemma.
Let $x_1, \ldots, x_{k'}$ with $k' \leq k$ be the points that were
connected to $v$ but are in different connected components than
$v$ after decreasing $v$'s power. This is because the only nodes
that might lose their connection to $v$ are within a distance
between $\rbig/\sqrt 3$ and $\rbig$, and there are at most $k$ such nodes
without a pairwise connection.

Consider $x_1$.
Let $z_0 = v$. According to Lemma~\ref{lem:iterclose}, there is a point
$z_1$ that is closer to $x_1$ and at most $2 \rball$ away from $v$.
Iteratively for $i = 1, 2, \ldots$, we distinguish three cases until this process stops:
\begin{enumerate}[(i)]
\setlength{\itemsep}{0mm}
\item $z_i$ belongs to the same component as $x_j$ for some $j$ ($z_i$ is closer to $x_1$ than $z_{i-1}$, but this does not
imply $j=1$).
  We increase $z_i$'s power such that $z_i$ is able to send to $z_{i-1}$.
If $i>1$, then we also increase $z_{i-1}$'s power accordingly.
\label{firstcase}
\item $z_i$ belongs to the same component as $v$. Then we can apply
Lemma~\ref{lem:iterclose} to $z_i$ and $x_1$
and find a point $z_{i+1}$
that is closer
to $x_1$ than $z_i$ and at most at a distance of $2\rball$ of $z_i$.
\item $z_i$ is within a distance of at most $2\rball$ of some $x_j$.
In this case, we increase the energy of $z_i$ such that $z_i$ and $x_j$ are connected. (The energy of $x_j$ is sufficiently large anyhow.)
\end{enumerate}
Running this process once decreases the number of connected
components by one and costs at most $2(2\rball)^p = 2^{p+1}
\rball^p$ additional power. We run this process $k' \leq k$ times,
thus spending at most $k2^{p+1} \rball^p$ of additional power. In
this way, we obtain a valid PA graph.

We have to show that the new PA graph indeed saves power. To do this,
we consider the power saved by decreasing $v$'s energy. By decreasing $v$'s
power, we save an amount of $\rbig^p - (\rbig/\sqrt 3)^p >
(1-\sqrt{3}^{-p}) \cdot \redge^p$. By the choice of $\cedge$,
the saved amount of energy exceeds the additional amount of $k2^{p+1}\rball^p$. This contradicts the optimality of the PA graph
with the edge of length $\rbig > \redge$.
\end{proof}

\begin{remark}
Since the longest edge has a length of at most $\redge$ with high
probability, i.e., with a probability of $1-n^{-\Omega(1)}$, and any ball of radius $\redge$ contains roughly
$O(\log n)$ points due to Chernoff's bound~\cite[Chapter
4]{MitzenmacherUpfal:ProbComp:2005}, the maximum degree of an
optimum PA graph of a random point set is $O(\log n)$ with high probability --
contrasting Lemma~\ref{lem:unbounded}.
\end{remark}

Yukich gave two different notions of smoothness in mean~\cite[(4.13) and (4.20) \& (4.21)]{Yukich:ProbEuclidean:1998}.
We use the stronger notion, which implies the other.

\begin{definition}[\mbox{smooth in mean~\cite[(4.20), (4.21)]{Yukich:ProbEuclidean:1998}}]
\label{def:smoothinmean}
A Euclidean functional $\functional$ is called \emph{smooth in mean}
if, for every constant $\beta > 0$, there exists a constant $c = c(\beta)$ such that
the following holds with a probability of at least $1-n^{-\beta}$:
\[
 \bigl|\functional(n) - \functional(n \pm k)\bigr|
\leq c k \cdot \bigl(\frac{\log n}{n}\bigr)^{p/d}
\]
and
\[
\bigl|\functional_B(n) - \functional_B(n \pm k)\bigr|
= c k \cdot \bigl(\frac{\log n}{n}\bigr)^{p/d} .
\]
for all $0 \leq k \leq n/2$.
\end{definition}

\begin{lemma}
\label{lem:smoothinmean}
$\pa_B$ and $\pa$ are smooth in mean for all $p > 0$ and all
$d$.
\end{lemma}

\begin{proof}
The bound
$\pa(n+k)
\leq \pa(n) + O\bigl(k \cdot \bigl(\frac{\log n}{n}\bigr)^{\frac pd}\bigr)$
follows from the fact that for all $k$ additional vertices, with a probability of
at least $1-n^{-\beta}$ for any $\beta > 0$, there is a vertex among the first $n$
within a distance of at most $O\bigl((\log n/n)^{1/d}\bigr)$ according to Lemma~\ref{lem:emptyball} ($\beta$ influences the constant hidden in the $O$).
Thus, we can connect any of the $k$ new vertices with costs
of $O\bigl((\log n/n)^{p/d}\bigr)$ to the optimal PA graph for the $n$ nodes.

Let us now show the reverse inequality
$\pa(n)
\leq \pa(n+k) + O\bigl(k \cdot \bigl(\frac{\log n}{n}\bigr)^{\frac pd}\bigr)$.
To do this, we show that with a probability of at least $1 - n^{-\beta}$, we have
\begin{equation}
\pa(n)
\leq \pa(n+1) + O\left(\left(\frac{\log n}{n}\right)^{\frac pd}\right) .
\label{equ:smoothiter}
\end{equation}
Then we iterate $k$ times to obtain the bound we aim for.

The proof of \eqref{equ:smoothiter} is similar to the analogous
inequality in Yukich's proof~\cite[Lemma
4.8]{Yukich:ProbEuclidean:1998}. The only difference is that we
first have to redistribute the power of the point $U_{n+1}$ to its
closest neighbors as in the proof of Lemma~\ref{lem:smooth}. In
this way, removing $U_{n+1}$ results in a constant number of
connected components. The longest edge incident to $U_{n+1}$ has a
length of $O\bigl((\log n/n)^{1/d}\bigr)$ with a probability
of at least $1-n^{-\beta}$ for any constant $\beta > 0$.
Thus, we can connect these constant number number of components
with extra power of at most $O\bigl((\log n/n)^{p/d}\bigr)$.

The proof of
\[
 \left|\pa(n) - \pa(n - k)\right|
= O\left(k \cdot \left(\frac{\log n}{n}\right)^{\frac pd}\right)
\]
and the statement
\[
 \left|\pa_B(n) - \pa_B(n \pm k)\right|
= O\left(k \cdot \left(\frac{\log n}{n}\right)^{\frac pd}\right)
\]
for the boundary functional are almost identical.
\end{proof}

\begin{definition}[\mbox{close in mean~\cite[(4.11)]{Yukich:ProbEuclidean:1998}}]
\label{def:closeinmean}
A Euclidean functional $\functional$ is close in mean to its boundary functional
$\functional_B$ if
\[
\expected\left(\left|\functional(n) - \functional_B(n)\right|\right)
   = o\left(n^{\stdexp}\right).
\]
\end{definition}

\begin{lemma}
\label{lem:closemean}
$\pa$ is close in mean to $\pa_B$ for all $d$ and $p \geq 1$.
\end{lemma}

\begin{proof}
It is clear that $\pa_B(X) \leq \pa(X)$ for all $X$. Thus, in what follows,
we prove that
$\pa(X) \leq \pa_B(X) + o\bigl(n^{\stdexp}\bigr)$ holds with a probability
of at least $1 - n^{-\beta}$, where $\beta$ influences the constant hidden
in the $o$.
This implies closeness in mean.

With a probability of at least $1-n^{-\beta}$, the longest edge in
the graph that realizes $\pa_B(X)$ has a length of $\cedge \cdot (\log
n/n)^{1/d}$ (Lemma~\ref{lem:longest}).
Thus, with a probability of at least $1 - n^{-\beta}$ for any constant $\beta > 0$,
only vertices within a distance of at
most $\cedge \cdot (\log n/n)^{1/d}$ of the boundary are connected to
the boundary. As the $d$-dimensional unit cube is bounded by $2^d$
hyperplanes, the expected number of vertices that are so close to
the boundary is bounded from above by $\cedge n 2^d \cdot (\log
n/n)^{1/d} = O\bigl((\log n)^{1/d} n^{\frac{d-1}d}\bigr)$. With
a probability of at least $1 - n^{-\beta}$ for any $\beta > 0$, this
number is exceeded by no more than a
constant factor.

Removing these vertices causes the boundary PA graph to fall into at most $O\bigl((\log n)^{1/d} n^{\frac{d-1}d}\bigr)$ components.
We choose one vertex of every component and start the process described in the proof of Lemma~\ref{lem:longest}
to connect all of them.
The costs per connection is bounded from above by $O\bigl((\log n/n)^{p/d}\bigr)$ with
a probability of $1 - n^{-\beta}$ for any constant $\beta > 0$.
Thus, the total costs are bounded from above by
\[
O\bigl((\log n/n)^{p/d}\bigr) \cdot O\bigl((\log n)^{1/d} n^{\frac{d-1}d}\bigr)
= O\left((\log n)^{\frac{p-1}d} \cdot n^{\frac{d-1-p}d}\right) = o\bigl(n^\stdexp\bigr)
\]
with a probability of at least $1 - n^{-\beta}$ for any constant $\beta > 0$.
\end{proof}

\section{Convergence}
\label{sec:convergence}

\subsection{Standard Convergence}
\label{ssec:standard}

Our findings of Sections~\ref{ssec:deterministic} yield
complete convergence of $\pa$ for $p<d$ (Theorem~\ref{thm:stdcc}).
Together with the probabilistic properties of Section~\ref{ssec:probabilistic},
we obtain convergence in mean in a straightforward way for all combinations of $d$ and $p$ (Theorem~\ref{thm:convmean}).
In Sections~\ref{ssec:warnke} and~\ref{ssec:cc}, we prove complete
convergence for $p \geq d$.

\begin{theorem}
\label{thm:stdcc}
For all $d$ and $p$ with $1 \leq p < d$, there exists a constant $\gamma_{\pa}^{d,p}$
such that
\[
 \frac{\pa^p(n)}{n^{\stdexp}}
\]
converges completely to $\gamma_{\pa}^{d,p}$.
\end{theorem}

\begin{proof}
This follows from the results in Section~\ref{ssec:deterministic} together
with results by Yukich~\cite[Theorem 4.1, Corollary 6.4]{Yukich:ProbEuclidean:1998}.
\end{proof}

\begin{theorem}
\label{thm:convmean}
For all $p \geq 1$ and $d \geq 1$, there exists a constant $\gamma_{\pa}^{d,p}$ (equal to the constant
of Theorem~\ref{thm:stdcc} for $p<d$)
such that
\[
\lim_{n \to \infty} \frac{\expected\bigl(\pa^p(n)\bigr)}{n^{\stdexp}}
=
\lim_{n \to \infty} \frac{\expected\bigl(\pa^p_B(n)\bigr)}{n^{\stdexp}}
= \gamma_{\pa}^{d,p}.
\]
\end{theorem}

\begin{proof}
This follows from the results in Sections~\ref{ssec:deterministic}
and~\ref{ssec:probabilistic} together
with results by Yukich~\cite[Theorem~4.5]{Yukich:ProbEuclidean:1998}.
\end{proof}

\subsection{Concentration with Warnke's Inequality}
\label{ssec:warnke}

McDiarmid's or Azuma-Hoeffding's inequality are powerful tools to prove concentration of measure
for a function that depends on many independent random variables, all of which have only a bounded influence
on the function value. If we consider smoothness in mean (see Lemma~\ref{lem:smoothinmean}), then we have the situation that
the influence of a single variable is typically very small (namely $O((\log n/n)^{p/d})$), but can be quite large
in the worst case (namely $O(1)$). Unfortunately, this situation is not covered by McDiarmid's or Azuma-Hoeffding's
inequality.
Fortunately, Warnke~\cite{Warnke} proved a generalization specifically for the case that the influence of single
variables is typically bounded and fulfills a weaker bound in the worst case.

The following theorem is a simplified version (personal communication with Lutz Warnke) of Warnke's concentration
inequality~\cite[Theorem 2]{Warnke}, tailored to our needs.

\begin{theorem}[Warnke]
\label{thm:TL}
Let $U_1, \ldots, U_n$ be a family of independent random variables with $U_i \in [0,1]^d$ for each $i$.
Suppose that there are numbers $\cgood \leq \cbad$ and an event $\Gamma$ such that the function $\functional:([0,1]^d)^n \to \real$ satisfies
\begin{multline}
 \max_{i \in [n]}\max_{x \in [0,1]^d}\left|\functional(U_1, \ldots, U_n)-\functional(U_1, \ldots, U_{i-1}, x, U_{i+1}, \ldots, U_k)\right| 
  \\ \leq \begin{cases}
        \cgood & \text{if $\Gamma$ holds and}\\
        \cbad & \text{otherwise.}
\end{cases}\label{eq:TL}
\end{multline}
Then, for any $t \geq 0$ and $\gamma \in (0,1]$ and $\eta = \gamma(\cbad-\cgood)$, we have
\begin{equation}\label{eq:PTL}
\textstyle
\probab\bigl(|\functional(n) - \expected(\functional(n))| \geq t\bigr) \le 2\exp\bigl(-\frac{t^2}{2n (\cgood+\eta)^2}\bigr) + \frac{n}{\gamma} \cdot
\probab(\neg \Gamma) .
\end{equation}
\end{theorem}

\begin{proof}[Proof sketch]
There are two differences of this simplified variant to Warnke's result~\cite[Theorem 2]{Warnke}:
First, the numbers $\cgood$ and $\cbad$ do not depend on the index $i$
but are chosen uniformly for all indices.
Second, and more importantly, the event $\mathcal B$~\cite[Theorem 2]{Warnke} is not used in Theorem~\ref{thm:TL}.
In Warnke's theorem~\cite[Theorem 2]{Warnke}, the event $\mathcal B$ plays only a bridging role:
it is required that $\probab(\mathcal B) \leq \sum_{i =1}^n \frac 1{\gamma_i} \cdot \probab(\lnot
\Gamma)$ for some $\gamma_1, \ldots, \gamma_n$ that show up in the tail bound as well.
Choosing $\gamma_i = \gamma$ for all $i$ yields
$\probab(\mathcal B) \leq \frac n{\gamma} \cdot \probab(\lnot \Gamma)$.
Then
\[
 \probab\bigl(\functional(n) \geq \expected(\functional(n) + t \text{ and } \lnot \mathcal B\bigr) \le \exp\left(-\frac{t^2}{2n (\cgood+\eta)^2}\right)
\]
yields
\[
\probab\bigl(|\functional(n) - \expected(\functional(n))| \geq t\bigr) \le 2\exp\left(-\frac{t^2}{2n (\cgood+\eta)^2}\right) + \frac{n}{\gamma} \cdot
\probab(\neg \Gamma)
\]
by observing that a two-sided tail bound can be obtained by symmetry and adding an upper bound for the probability of $\mathcal B$ to the
right-hand side.
\end{proof}

Next, we define \emph{typical smoothness}, which means that, with high probability,
a single point does not have a significant influence on the value of~$\functional$, and we apply
Theorem~\ref{thm:TL} for typically smooth functionals $\functional$.
The bound of $c \cdot (\log n/n)^{p/d}$ in Definition~\ref{def:typsmooth} below for the typical influence of a single point is somewhat arbitrary, but works
for $\pa$ and $\mst$. This bound is also essentially the smallest possible, as for there can be regions of diameter $c' \cdot (\log n/n)^{1/d}$ for some small
constant $c' > 0$ that contain no or only a single point.
It might be possible to obtain convergence results for other functionals
for weaker notions of typical smoothness.

\begin{definition}[typically smooth]
\label{def:typsmooth}
A Euclidean functional $\functional$ is \emph{typically smooth}
if, for every $\beta > 0$, there exists a constant $c = c(\beta)$ such that
\[
  \max_{x \in [0,1]^d, i \in [n]} \bigl|\functional(U_1, \ldots, U_n)
 - \functional(U_1, \ldots, U_{i-1}, x, U_{i+1}, \ldots, U_n) \bigr| \leq 
  c \cdot \left(\frac{\log n}{n}\right)^{p/d}
\]
with a probability of at least $1 - n^{-\beta}$.
\end{definition}

\begin{theorem}[concentration of typically smooth functionals]
\label{thm:concentration}
Assume that $\functional$ is typically smooth.
Then
\[
\probab\bigl(|\functional(n) - \expected(\functional(n))| \geq t\bigr)
\leq O(n^{-\beta}) + \exp\left(- \frac{t^2 n^{\frac{2p}d -1}}{C (\log n)^{2p/d}}\right)
\]
for an arbitrarily large constant $\beta> 0$ and another constant $C>0$
that depends on $\beta$.
\end{theorem}

\begin{proof}
We use Theorem~\ref{thm:TL}.
The event $\Gamma$ is that any point can change the value only by
at most $O\bigl((\log n/n)^{p/d})$.
Thus, $\cgood = O\bigl((\log n/n)^{p/d})$ and $\cbad = O(1)$.
The probability that we do not have the event $\Gamma$ is bounded by $O(n^{-\beta})$
for an arbitrarily large constant $\beta$ by typical smoothness.
This only influences the constant
hidden in the $O$ of the definition of $\cgood$.

We choose $\gamma = O\bigl((\log n/n)^{p/d})$.
In the notation of Theorem~\ref{thm:TL}, we choose $\eta = O(\gamma)$, which is possible as
$\cbad - \cgood \approx \cbad = \Theta(1)$.
Using the conclusion of Theorem~\ref{thm:TL} yields
\begin{align*}
\textstyle \probab\bigl(|\functional(n)  - \expected(\functional(n)) |  \geq t\bigr)
& \textstyle\leq \frac n \gamma \cdot \probab(\lnot \Gamma) + \exp\left(- \frac{t^2 n^{2p/d}}{nC (\log n)^{2p/d}}\right) \\
& \textstyle \leq O(n^{-\beta})+ \exp\left(- \frac{t^2 n^{2p/d}}{nC (\log n)^{2p/d}}\right)
\end{align*}
for some constant $C > 0$.
Here, $\beta$ can be chosen arbitrarily large.
\end{proof}

Choosing $t = n^\stdexp/\log n$ yields a nontrivial concentration
result that suffices to prove complete convergence of typically smooth Euclidean functionals.

\begin{corollary}
\label{cor:tail}
Assume that $\functional$ is typically smooth. Then
\begin{equation}
\probab\bigl(|\functional(n) - \expected(\functional(n))| > n^\stdexp/\log n\bigr)
     \leq O\left(n^{-\beta}  + \exp\left(- \frac{n}{C (\log n)^{2 + \frac{2p}d}}\right) \right) \label{equ:explicitconcentration}
\end{equation}
for any constant $\beta$ and $C$ depending on $\beta$ as in Theorem~\ref{thm:concentration}.
\end{corollary}

\begin{proof}
The proof is straightforward by exploiting that the assumption that $\functional(n)/n^\stdexp$ converges in mean to $\gamma_{\functional}^{d,p}$
implies $\expected(\functional(n)) = \Theta(n^\stdexp)$.
\end{proof}

\subsection[Complete Convergence for p>=d]{\boldmath Complete Convergence for $p \geq d$}
\label{ssec:cc}
In this section, we prove that typical smoothness (Definition~\ref{def:typsmooth}) suffices for complete convergence.
This implies complete convergence of $\mst$ and $\pa$ by Lemma~\ref{lem:typsmooth} below.

\begin{theorem}
\label{thm:cc}
Assume that $\functional$ is typically smooth and
$\functional(n)/n^{\stdexp}$ converges in mean to $\gamma_{\functional}^{d,p}$.
Then $\functional(n)/n^{\stdexp}$ converges completely to $\gamma_{\functional}^{d,p}$.
\end{theorem}

\begin{proof}
Fix any $\eps > 0$. Since
\[
  \lim_{n \to \infty} \expected\left(\frac{\functional(n)}{n^{\stdexp}}\right) = \gamma_{\functional}^{d,p} ,
\]
there exists an $n_0$
such that
\[
  \expected\left(\frac{\functional(n)}{n^{\stdexp}}\right) \in \left[\gamma_{\functional}^{d,p} - \frac \eps 2 ,\gamma_{\functional}^{d,p} + \frac \eps 2 \right]
\]
for all $n \geq n_0$.

Furthermore, there exists an $n_1$ such that, for
all $n \geq n_1$, the probability that $\functional(n)/n^\stdexp$
deviates by more than $\eps/2$ from its expected value is smaller
than $n^{-2}$ for all $n \geq n_1$. To see this, we use
Corollary~\ref{cor:tail} and observe that the right-hand side of
\eqref{equ:explicitconcentration} is $O(n^{-2})$ for sufficiently
large $\beta$ and that the event on the left-hand side is
equivalent to
\[
\left|\frac{\functional(n)}{n^\stdexp} - \frac{\expected(\functional(n))}{n^\stdexp}\right| > O\left(\frac 1{\log n}\right),
\]
where $O(1/\log n) < \eps/2$ for sufficiently large $n_1$ and $n \geq n_1$.
Let $n_2 = \max\{n_0, n_1\}$. Then
\[
\sum_{n=1}^\infty \probab\left(\left|\frac{\pa(X)}{n^{\stdexp}}\right| > \eps \right) \leq n_2 + \sum_{n=n_2+1}^\infty n^{-2}
 = n_2 + O(1) < \infty.
\]
\end{proof}

Although similar in flavor, smoothness in mean does not immediately imply typical smoothness or vice versa:
the latter makes only a statement about \emph{single} points at \emph{worst-case} positions. The former only makes a statement about adding and removing
\emph{several} points at \emph{random} positions. However, the proofs of smoothness in mean for $\mst$ and $\pa$ do not exploit this, and we can adapt them
to yield typical smoothness.

\begin{lemma}
\label{lem:typsmooth}
$\pa$ and $\mst$ are typically smooth.
\end{lemma}

\begin{proof}
We first consider $\pa$. Replacing a point $U_k$ by some other (worst-case) point $z$ can be modeled by removing $U_k$
and adding $z$. We observe that, in the proof of smoothness in mean (Lemma~\ref{lem:smoothinmean},
we did not exploit that the point added is at a random position, but the proof goes through for any single point at an arbitrary position.
Also the other way around, i.e., removing $z$ and replacing it by a random point $U_k$, works in the same way.
Thus, $\pa$ is typically smooth.

Closely examining Yukich's proof of smoothness in mean for
$\mst$~\cite[Lemma 4.8]{Yukich:ProbEuclidean:1998} yields the same
result for $\mst$.
\end{proof}

\begin{corollary}
\label{cor:ccexample}
For all $d$ and $p$ with $p \geq 1$,
$\mst(n)/n^\stdexp$ and $\pa(n)/n^\stdexp$
converge completely
to constants $\gamma_{\mst}^{d,p}$ and $\gamma_{\pa}^{d,p}$, respectively.
\end{corollary}

\begin{proof}
Both $\mst$ and $\pa$ are typically smooth and converge in mean. Thus, the corollary follows
from Theorem~\ref{thm:cc}.
\end{proof}

\begin{remark}
Instead of Warnke's method of typical bounded differences, we could also have used Kutin's extension of McDiarmid's inequality~\cite[Chapter 3]{Kutin:PhD:2002}.
However, this inequality yields only convergence for $p \leq 2d$, which is still an improvement over the previous complete convergence of $p<d$,
but weaker than what we get with Warnke's inequality. Furthermore, Warnke's inequality is easier to apply
and a more natural extension in the following way: intuitively, one might think that we could just take
McDiarmid's inequality and add the probability that we are not in a nice situation using a simple union bound, but, in general, this is not true~\cite[Section 2.2]{Warnke}.
\end{remark}

\section{Average-Case Approximation Ratio of the MST Heuristic}
\label{sec:MST}

In this section, we show that the average-case approximation
ratio of the MST heuristic for power assignments is strictly better than its worst-case ratio of $2$. First, we prove
that the average-case bound is strictly (albeit marginally) better than $2$ for any combination of $d$ and $p$.
Second, we show a simple improved bound for the 1-dimensional case.

\subsection{The General Case}
\label{ssec:generalcase}

The idea behind showing that the MST heuristic performs better on average than in the worst case is as follows:
the weight of the PA graph obtained from the MST heuristic can not only be upper-bounded by twice the weight of an MST,
but it is in fact easy to prove that it can be upper-bounded by twice the weight of the heavier half of the edges of the MST~\cite{AveragePA}.
Thus, we only have to show that the lighter half of the edges of the MST contributes $\Omega(n^\stdexp)$ to the value of the MST in expectation.

For simplicity, we assume that the number $n=2m+1$ of points is odd.
The case of even $n$ is similar but slightly more technical.
We draw points $X=\{U_1, \ldots, U_n\}$ as described above.
Let $\pt(X)$ denote the power required in the
power assignment obtained from the MST. Furthermore, let $\heavy$ denote
the $m$ heaviest edges of the MST, and let $\light$ denote the $m$
lightest edges of the MST. We omit the parameter $X$ since it is clear from the
context.
Then we have
\begin{equation}
\label{equ:treerelation}
 \heavy + \light = \mst \leq \pa \leq \pt \leq 2 \heavy = 2 \mst
  - 2\light \leq 2\mst
\end{equation}
since the weight of the PA graph obtained from an MST can not only be upper bounded by
twice the weight of a minimum-weight spanning tree, but it is easy to show that the PA graph obtained from the MST is in fact by twice the weight of the heavier half of the edges
of a minimum-weight spanning tree~\cite{AveragePA}.

For distances raised to the power $p$, the expected value of $\mst$ is
$(\gamma_{\mst}^{d,p} \pm o(1)) \cdot n^{\stdexp}$.
If we can prove that the lightest $m$ edges of the MST are of weight
$\Omega(n^{\stdexp})$, then
it follows that the MST power assignment is strictly less than twice the optimal power
assignment.
$\light$ is lower-bounded by the weight of the lightest $m$ edges of the whole graph
without any further constraints. Let $\reallight = \reallight(X)$ denote the weight
of these $m$ lightest edges of the whole graph. Note that both $\light$ and $\reallight$
take edge lengths to the $p$-power, and we have $\reallight \leq \light$.

Let $c$ be a small constant to be specified later on.
Let $v_{d,r} = \frac{\pi^{d/2} r^d}{\Gamma(\frac n2 +1)}$ be the volume of a
$d$-dimensional ball of radius $r$. For compactness,
we abbreviate $c_d = \frac{\pi^{d/2}}{\Gamma(\frac n2 +1)}$, thus
$v_{d, r} = c_d r^d$. Note that all $c_d$'s are constants since
$d$ is constant.

The probability $P_k$ that a fixed vertex $v$ has at least $k$ other vertices within a distance
of at most $r=\ell \cdot \sqrt[d]{1/n}$ for some constant $\ell > 0$ is bounded from above
by
\[
   P_k \leq \binom{n-1}k \cdot v_{d,r}^k
   \leq \frac{n^k (c_dr^d)^k}{k!}
   =\frac{n^k (c_d \ell^d n^{-1})^k}{k!}
   = \frac{\tilde c^k}{k!}
\]
for another constant $\tilde c = \ell^d c_d$.
This follows from independence and a union bound.
The expected number of edges of a specific vertex that have a length of at most
$r$ is thus bounded from above by
\[
   \sum_{k=1}^{n-1} P_k \leq \sum_{k=1}^{n-1} \frac{\tilde c^k}{k!}
   \leq \sum_{k=1}^\infty \frac{\tilde c^k}{k!} = e^{\tilde c} -1.
\]
By choosing $\ell$ appropriately small, we can achieve that $\tilde c \leq 1/3$.
This yields $e^{\tilde c} - 1 < 1/2$.
By linearity of expectation, the total number of edges of length at most $r$
in the whole graph is bounded from above by $m/2$.
Thus, at least $m/2$ of the lightest $m$ edges of the whole graph have a length of at least $r$.
Hence, the expected value of $\reallight$ is bounded from below by
\[
  \frac m2 \cdot r^p = \frac m2 \cdot \ell^p n^{-\frac pd} \leq
\frac{\ell^p}4 \cdot n^{\stdexp} = C_{\reallight}^{d,p} \cdot n^{\stdexp} .
\]
for some constant $C_{\reallight}^{d,p} > 0$.
Then the expected value of $\pt$ is bounded from above by
\[
  \left(2\gamma_{\mst}^{d,p} - 2 C_{\reallight}^{d,p} + o(1)\right) \cdot n^{\stdexp}
\]
by \eqref{equ:treerelation}.
From this and the convergence of $\pa$, we can conclude the following theorem.

\begin{theorem}
\label{thm:mstratio}
For any $d \geq 1$ and any $p \geq 1$, we have
\[
\gamma_{\mst}^{d,p} \leq \gamma_{\pa}^{d,p} \leq
2 \gamma_{\mst}^{d,p} - 2 C_{\reallight}^{d,p} < 2 \gamma_{\mst}^{d,p}
\]
for some constant $C_{\reallight}^{d,p} > 0$ that depends only on $d$ and $p$.
\end{theorem}

By exploiting that in particular $\pa$ converges completely, we can obtain
a bound on the expected approximation ratio from the above result.

\begin{corollary}
\label{cor:mstratio}
For any $d \geq 1$ and $p \geq 1$ and sufficiently large $n$,
the expected approximation ratio of the MST heuristic for power assignments
is bounded from above by a constant strictly smaller than $2$.
\end{corollary}

\begin{proof}
The expected approximation ratio is $\expected\bigl(\pt(n)/\pa(n)\bigr)
= \expected\bigl(\frac{\pt(n)/n^\stdexp}{\pa(n)/n^\stdexp}\bigr)$.
We know that $\pa(n)/n^\stdexp$ converges completely to $\gamma_{\pa}^{d,p}$.
This implies that the probability that $\pa(n)/n^\stdexp$ deviates by more than $\eps>0$
from $\gamma_{\pa}^{d,p}$ is $o(1)$ for any $\eps > 0$.

If $\pa(n)/n^\stdexp \in [\gamma_{\pa}^{d,p} - \eps, \gamma_{\pa}^{d,p} + \eps]$,
then the expected approximation ratio can be bounded from above by $\frac{2 \gamma_{\mst}^{d,p} - 2 C_{\reallight}^{d,p}}{\gamma_{\pa}^{d,p} - \eps}$.
This is strictly smaller than $2$ for a sufficiently small $\eps > 0$.

Otherwise, we bound the expected approximation ratio by the worst-case ratio of $2$, which contributes only $o(1)$ to its expected value.
\end{proof}

\begin{remark}
\label{rem:ptconv}
Complete convergence of the functional $\pt$ as well as smoothness and closeness in mean has been shown for the specific case $p=d$~\cite{AveragePA}.
We believe that $\pt$ converges completely for all $p$ and $d$.
Since then $\gamma_{\pt}^{d,p} \leq 2 \gamma_{\mst}^{d,p} - 2 C_{\reallight}^{d,p} <
2 \gamma_{\mst}^{d,p}$, we would obtain
a simpler proof of Corollary~\ref{cor:mstratio}.
\end{remark}

\subsection{An Improved Bound for the One-Dimensional Case}

The case $d=1$ is much simpler than the general case, because the MST is just a Hamiltonian path starting at the left-most and ending
at the right-most point. Furthermore, we also know precisely what the MST heuristic does:
assume that a point $x_i$ lies between $x_{i-1}$ and $x_{i+1}$. The MST heuristic assigns power
$\pa(x_i) = \max\{|x_{i} - x_{i-1}|, |x_{i} - x_{i+1}|\}^p$ to $x_i$. The example that proves that the MST heuristic is
no better than a worst-case 2-approximation shows that it is bad if $x_i$ is very close to either side and good if $x_i$
is approximately in the middle between $x_{i-1}$ and $x_{i+1}$.

In order to show an improved bound for the approximation ratio of the MST heuristic for $d=1$, we introduce some notation. First we remark that for $X=\{U_1,\ldots, U_n\} $ with high probability, there is no subinterval
of length $c \log n/n$ of $[0,1]$ that does not contain any of the $n$ points
$U_1, \ldots, U_n$ (see Lemma~\ref{lem:emptyball} for the precise statement).

We assume that no interval of length $c \log n/n$ is empty
for some sufficiently large constant $c$ for the rest of this section.

We proceed as follows:
Let $x_0 = 0$, $x_{n+1} = 1$, and let $x_1 \leq \ldots \leq x_n$ be the $n$ points (sorted in increasing
order) that are drawn uniformly and independently from the interval $[0,1]$.

Now we distribute the weight of the power assignment $\pt(X)$ in the
power assignment obtained from the MST, and the
weight of the MST as follows: For the power
assignment, every point $x_i$ (for $1 \leq i \leq n$) gets a
charge of $P_i = \max\{x_i - x_{i-1}, x_{i+1}-x_i\}^p$. This is
precisely the power that this point needs in the power assignment
obtained from the spanning tree. For the minimum spanning tree, we
divide the power of an edge $(x_{i-1}, x_i)$ (for $1 \leq i \leq
n+1$) evenly between $x_{i-1}$ and $x_i$. This means that the
charge of $x_i$ is $M_i = \frac 12 \cdot \bigl((x_i - x_{i-1})^p +
(x_{i+1}-x_i)^{p}\bigr)$.

The length of the minimum spanning tree is thus
\[
  \mst =\underbrace{\sum_{i=1}^n M_i}_{M^\star} + \underbrace{\frac 12 \cdot \bigl((x_1-x_0)^p + (x_{n+1} - x_n)^p\bigr)}_{=M'}.
\]
The total power for the power assignment obtained from this tree
is
\[
  \pt = \underbrace{\sum_{i=1}^n P_i}_{P^\star} + \underbrace{(x_1-x_0)^p + (x_{n+1} - x_n)^p}_{=P'}.
\]
Note the following: If the largest empty interval has a length of at most $c \log n/n$, then
the terms $P'$ and $M'$ are negligible according to the following lemma.
Thus, we ignore $P'$ and $M'$ afterwards to simplify the analysis.

\begin{lemma}
\label{lem:negligible}
Assume that the largest empty interval has a length of at most $c \log n/n$.
Then $M' = O\bigl(M^\star \cdot \frac{(\log n)^p}{n}\bigr)$
and $P' = O\bigl(P^\star \cdot \frac{(\log n)^p}{n}\bigr)$.
\end{lemma}

\begin{proof}
We have $M' \leq (c \log n/n)^p$ and $P' \leq 2 (c \log n/n)^p$ because
$x_1 \leq c \log n/n$ and $x_n \geq 1-c \log n/n$ by assumption.
Thus, $M', P' = O\bigl(\bigl(\frac{\log n}n\bigr)^p\bigr)$.
Furthermore,
\[
  M^\star = \sum_{i=1}^n \frac 12 \cdot \left(\bigl(x_i - x_{i-1}\bigr)^p + \bigl(x_{i+1}-x_i\bigr)^p\right).
\]
Since $p \geq 1$, this function becomes minimal if we place $x_1, \ldots, x_n$ equidistantly.
Thus,
\[
  M^\star \geq \sum_{i=1}^n \left(\frac 1{n+1}\right)^p = n \cdot \left(\frac 1{n+1}\right)^p
   = \Omega\bigl(n^{1-p}\bigr).
\]
With a similar calculation, we obtain $P^\star = \Omega\bigl(n^{1-p}\bigr)$
and the result follows.
\end{proof}

For simplicity, we assume from now on that $n$ is even. If $n$ is odd, the proof proceeds
in exactly the same way except for some changes in the indices.
In order to analyze $M$ and $P$, we proceed in two steps:
First, we draw all points $x_1, x_3, \ldots, x_{n-1}$ (called the \emph{odd points}). Given the locations of these
points, $x_i$ for even $i$ ($x_i$ is then called an \emph{even point})
is distributed uniformly in the interval
$[x_{i-1}, x_{i+1}]$.
Note that we do not really draw the odd points. Instead, we let an adversary fix these points.
But the adversary is not allowed to keep an interval of length $c \log n/n$ free (because
randomness would not do so either with high probability).
Then the sums
\[
 \meven = \sum_{i=1}^{n/2} M_{2i}
\]
and
\[
  \peven = \sum_{i=1}^{n/2} P_{2i}
\]
are sums of independent random variables. (Of course $M_{2i}$ and $P_{2i}$ are dependent.)
Now let $\ell_{2i} = x_{2i+1} - x_{2i-1}$ be the length of the interval for $x_{2i}$.
The expected value of $M_{2i}$ is
\[
  \expected(M_{2i}) = \frac 1{\ell_{2i}} \cdot \int_{0}^{\ell_{2i}} \frac 12 \cdot \bigl(x^p + (\ell_{2i} - x)^p\bigr)
   \,\text dx
   = \frac{\ell_{2i}^p}{p+1}.
\]
Analogously, we obtain
\begin{align*}
   \expected(P_{2i}) & = \frac 1{\ell_{2i}} \cdot \int_{0}^{\ell_{2i}} \max\{x, \ell_{2i} - x\}^p \,\text dx \\
   & = \frac 2{\ell_{2i}} \cdot \int_{0}^{\ell_{2i}/2} (\ell_{2i} - x)^p \,\text dx
   = \left(2-\frac{1}{2^p}\right) \cdot \frac{\ell_{2i}^p}{p+1} .
\end{align*}

We observe that $\expected(P_{2i})$ is a factor $2-2^{-p}$ greater than $\expected(M_{2i})$. In the same way,
the expected value of $P_{2i+1}$ is a factor of $2-2^{-p}$ greater than the expected value of $M_{2i+1}$.
This is already an indicator that the approximation ratio should be $2-2^{-p}$.

Because $\meven$ and $\peven$ are sums of independent random variables, we can
use Hoeffding's inequality to bound the probability that they deviate from the expected values
$\expected(\meven)$ and $\expected(\peven)$.

\begin{lemma}[Hoeffding's inequality~\cite{Hoeffding:SumsBounded:1963}]
\label{lem:hoeffding}
  Let $X_1, \ldots, X_m$ be independent random variables, where $X_i$ assumes
  values in the interval $[a_i, b_i]$. Let $X = \sum_{i = 1}^m X_i$. Then for
  all $t > 0$,
  \[
         \probab\bigl(X - \expected(X) \geq t\bigr)
    \leq \exp\left(-\frac{2 t^2}{\sum_{i = 1}^m (b_i-a_i)^2}\right).
  \]
  By symmetry, the same bound holds for
  $\probab\bigl(X - \expected(X) \leq -t\bigr)$.
\end{lemma}

Let us start with analyzing the probability that
$\meven < (1-n^{-1/4}) \cdot \expected(\meven)$.
We have $m=n/2$ in the above.
Furthermore, we have $b_i = \ell_{2i}^p/2$
(obtained if $x_{2i} = x_{2i-1}$ or $x_{2i} = x_{2i+1}$)
and $a_i = (\ell_{2i}/2)^p$.
Thus, $(b_i - a_i)^2 = \ell_{2i}^{2p} \cdot (2^{-1} - 2^{-p})^2$.
If $p > 1$ is a constant, then this is $c_{p} \ell_{2i}^{2p}$
for some constant $c_p$.
For $p = 1$, it is $0$. However, in this case, the length of the minimum spanning tree
is exactly $1$, without any randomness. Thus, for $p = 1$, we do not have to apply Hoeffding's inequality.

For $p > 1$, we obtain
\begin{align}
\probab\left(\meven < \bigl(1-n^{-1/4}\bigr) \cdot \expected(\meven)\right)
& \leq \exp\left( - \frac{2 n^{-1/2} \expected(\meven)^2}{\sum_{i=1}^{n/2}c_{p} \ell_{2i}^{2p}}\right) \notag \\
& = \exp\left( - \frac{2 n^{-1/2} \bigl(\sum_{i=1}^{n/2} \frac{\ell_{2i}^p}{p+1}\bigr)^2}{\sum_{i=1}^{n/2}c_{p} \ell_{2i}^{2p}}\right) \notag \\
& = \exp\left( - c' n^{-1/2} \cdot \frac{\bigl(\sum_{i=1}^{n/2} \ell_{2i}^p\bigr)^2}{\sum_{i=1}^{n/2} \ell_{2i}^{2p}}\right)
\label{fraction}
\end{align}
with $c' = \frac{2}{(p+1)^2 c_p}$.
To estimate the exponent, we use the following technical lemma.

\begin{lemma}
\label{lem:tech}
Let $p \geq 1$ be a constant.
Let $s_1, \ldots, s_m \in [0, \beta]$ be positive numbers for some $\beta > 0$
with $\sum_{i=1}^m s_i = \gamma$ for some number $\gamma$. (We assume that $m\beta \geq \gamma$.)
Then
\[
  \frac{\bigl(\sum_{i=1}^m s_i^p\bigr)^2}{\sum_{i=1}^m s_i^{2p}} \geq m \cdot \left(\frac{\gamma}{m\beta}\right)^p.
\]
\end{lemma}
\begin{proof}
We rewrite the numerator as
\[
\sum_{i=1}^m s_i^p \sum_{j=1}^m s_j^p
\]
and the denominator as
\[
\sum_{i=1}^m s_i^p s_i^p .
\]
Now we see that the coefficient for $s_i^p$ in the numerator is
$\sum_{j=1}^m s_j^p$, and it is $s_i^p \leq \beta^p$ in the denominator.
Because of $p \geq 1$, the sum $\sum_{j=1}^m s_j^p$ is convex as a function
of the $s_j$. Thus, it becomes minimal if all $s_j$ are equal.
Thus, the numerator is bounded from below
by $m \cdot (\gamma/m)^p$.
\end{proof}

With these results, we obtain the following theorem.

\begin{theorem}
\label{thm:d1}
For all $p \geq 1$, we have $\gamma_{\mst}^{1,p} \leq \gamma_{\pa}^{1,p} \leq (2 - 2^{-p}) \cdot \gamma_{\mst}^{1,p}$.
\end{theorem}

\begin{proof}
The first inequality is immediate.
For the second inequality, we apply Lemma~\ref{lem:tech} with $\beta = \frac{4\log n}{n}$ and $\gamma = 1-o(1) \geq 1/2$ (the $o(1)$ stems
from the fact the we have to ignore the distance $x_1 - x_0$ and even $x_{n+1} - x_n$)
and $s_i = \ell_{2i}$
and $m=n/2$ to obtain a lower bound of
$\frac n2 \cdot \left(\frac{1}{4\log n}\right)^p$
for the ratio of the fraction in \eqref{fraction}.

This yields
\[
\probab\left(\meven < \bigl(1-n^{-1/4}\bigr) \cdot \expected(\meven)\right) \leq
\exp\left(-\Omega\left(\frac{\sqrt n}{(\log n)^p}\right)\right).
\]
In the same way, we can show that
\[
\probab\left(\peven > \bigl(1+n^{-1/4}\bigr) \cdot \expected(\peven)\right) \leq \exp\left( - \Omega(n^{1/4})\right).
\]
Furthermore, the same analysis can be done for $\podd$ and $\modd$.

Thus, both the power assignment obtained from the MST
and the MST are concentrated around their means, their means
are at a factor of $2- 2^{-p}$ for large $n$, and the MST provides
a lower bound for the optimal PA.
\end{proof}

The high probability bounds for the bound of $2-2^{-p}$ of the approximation ratio of the power assignment
obtained from the spanning tree together with the observation that in case of any ``failure''
event we can use the worst-case approximation ratio of $2$ yields the following corollary.

\begin{corollary}
\label{cor:d1}
The expected approximation ratio of the MST heuristic is at most $2-2^{-p} + o(1)$.
\end{corollary}

\section{Conclusions and Open Problems}
\label{sec:concl}
We have proved complete convergence of Euclidean functionals that are \emph{typically smooth} (Definition~\ref{def:typsmooth})
for the case that the power $p$ is larger than the dimension $d$. The case $p > d$ appears naturally in the case of transmission questions for
wireless networks.

As examples, we have obtained complete convergence for the $\mst$ (minimum-spanning tree) and the $\pa$ (power assignment)
functional.
To prove this, we have used a recent concentration of measure result by Warnke~\cite{Warnke}.
His strong concentration inequality might be of independent interest to the algorithms community.
As a technical challenge, we have had to deal with the fact that the degree of an optimal power assignment graph can be unbounded.

To conclude this paper, let us mention some problems for further research:
\begin{enumerate}
\setlength{\itemsep}{0mm}
\item Is it possible to prove complete convergence of other functionals for $p \geq d$? The most prominent
one would be the traveling salesman problem (TSP). However, we are not aware that the TSP is smooth in mean,
\item Concerning the average-case approximation ratio of the MST heuristic, we only proved that the approximation
ratio is smaller than $2$. Only for the case $d=1$, we provided an explicit upper bound for the approximation ratio.
   Is it possible to provide an improved approximation ratio as a function of $d$ and $p$ for general $d$?
\item Can Rhee's isoperimetric inequality~\cite{Rhee:Subadditive:1993} be adapted to work for $p \geq d$?
   Rhee's inequality can be used to obtain convergence for the case that the points are not identically distributed, and has for instance been
   used for a smoothed analysis of Euclidean functionals~\cite{BlaeserEA:Partitioning:2013}. (Smoothed analysis has been
   introduced by Spielman and Teng to explain the performance of the simplex method~\cite{SpielmanTeng:SmoothedAnalysisWhy:2004}.
   We refer to two surveys for an overview~\cite{SpielmanTeng:CACM:2009,MantheyRoeglin:SmoothedSurvey:2011}.)
\item Can our findings about power assignments be generalized to other settings? For instance, to get a more reliable network, we may want to
have higher connectivity. Another issue would be to take into account interference of signals or noise such as the SINR or related models.
\end{enumerate}

\section*{Acknowledgment}
We thank Samuel Kutin, Lutz Warnke, and Joseph Yukich for fruitful discussions.


\end{document}